\documentclass[11pt]{article}
\newcommand{\fA}{\mathcal{A}}

\usepackage[letterpaper,margin=1.00in]{geometry}
\usepackage{amsmath, amssymb, amsthm, amsfonts}
\usepackage{bbm}

\usepackage{ifthen}
\usepackage{tikz}
\usetikzlibrary{positioning,decorations.pathreplacing}
\usepackage{cite}
\usepackage{appendix}
\usepackage{graphicx}
\usepackage{color}
\usepackage{color}
\usepackage{algorithm}
\usepackage{algorithm}
\usepackage[noend]{algpseudocode}
\usepackage{epstopdf}
\usepackage{wrapfig}
\usepackage{paralist}
\usepackage{wasysym}
\usepackage{comment}
\usepackage[textsize=tiny]{todonotes}

\usepackage{framed}
\usepackage[framemethod=tikz]{mdframed}
\usepackage[bottom]{footmisc}
\usepackage{enumitem}
\setitemize{noitemsep,topsep=3pt,parsep=3pt,partopsep=3pt}
\usepackage[font=small]{caption}
\usepackage{xspace}
\usepackage{thmtools}
\usepackage{thm-restate}

\newtheorem{theorem}{Theorem}[section]
\newtheorem{lemma}[theorem]{Lemma}
\newtheorem{meta-theorem}[theorem]{Meta-Theorem}

\newtheorem{remark}[theorem]{Remark}

\newtheorem{conjecture}[theorem]{Conjecture}
\newtheorem{observation}[theorem]{Observation}
\newtheorem{definition}[theorem]{Definition}

\newtheorem*{claim*}{Claim}

\definecolor{darkgreen}{rgb}{0,0.5,0}
\usepackage{hyperref}
\hypersetup{
    unicode=false,          
    colorlinks=true,        
    linkcolor=red,          
    citecolor=darkgreen,        
    filecolor=magenta,      
    urlcolor=cyan           
}
\usepackage[capitalize, nameinlink]{cleveref}
\crefname{theorem}{Theorem}{Theorems}
\crefname{problem}{Problem}{Problems}
\crefname{proposition}{Proposition}{Propositions}
\crefname{observation}{Observation}{Observations}
\crefname{conjecture}{Conjecture}{Conjectures}
\Crefname{lemma}{Lemma}{Lemmas}

\algnewcommand\algorithmicswitch{\textbf{switch}}
\algnewcommand\algorithmiccase{\textbf{case}}

\algdef{SE}[SWITCH]{Switch}{EndSwitch}[1]{\algorithmicswitch\ #1\ \algorithmicdo}{\algorithmicend\ \algorithmicswitch}%
\algdef{SE}[CASE]{Case}{EndCase}[1]{\algorithmiccase\ #1}{\algorithmicend\ \algorithmiccase}%
\algtext*{EndSwitch}%
\algtext*{EndCase}%

\newcommand{\poly}{\operatorname{poly}}

\newcommand{\eps}{\varepsilon}

\newcommand{\fG}{\mathcal{G}}

\newcommand{\fP}{\mathcal{P}}

\newcommand{\girth}{\mathrm{girth}}

\newcommand{\local}{$\mathsf{LOCAL}$\xspace}

\newcommand{\lca}{$\mathsf{LCA}$\xspace}
\newcommand{\volume}{$\mathsf{VOLUME}$\xspace}
\newcommand{\mpc}{$\mathsf{MPC}$\xspace}

\DeclareMathOperator{\tup}{tup}
\DeclareMathOperator{\ind}{ind}
\DeclareMathOperator{\seq}{seq}
\DeclareMathOperator{\dett}{det}

\renewcommand{\paragraph}[1]{\vspace{0.15cm}\noindent {\bf #1}:}

\newcommand{\FullOrShort}{full}

\ifthenelse{\equal{\FullOrShort}{full}}{
  
  \newcommand{\fullOnly}[1]{#1}
  \newcommand{\shortOnly}[1]{}

  }{

    \newcommand{\fullOnly}[1]{}
    \newcommand{\IncludePictures}[1]{}
   
  }

\begin{document}

\title{The Randomized Local Computation Complexity\\of the Lovász Local Lemma}

\author{Sebastian Brandt\\ ETH Zurich \and Christoph Grunau\\ ETH Zurich \and V\'aclav Rozho\v{n}\\ ETH Zurich}

\maketitle

\begin{abstract}
    The Local Computation Algorithm (\lca) model is a popular model in the field of sublinear-time algorithms that measures the complexity of an algorithm by the number of probes the algorithm makes in the neighborhood of one node to determine that node's output. 
    
   In this paper we show that the randomized \lca complexity of the Lovász Local Lemma (LLL) on constant degree graphs is $\Theta(\log n)$. 
   The lower bound follows by proving an $\Omega(\log n)$ lower bound for the Sinkless Orientation problem introduced in [Brandt et al.\ STOC 2016]. This answers a question of [Rosenbaum, Suomela PODC 2020].
   
   Additionally, we show that every randomized \lca algorithm for a locally checkable problem with a probe complexity of $o(\sqrt{\log{n}})$ can be turned into a deterministic \lca algorithm with a probe complexity of $O(\log^* n)$. This improves exponentially upon the currently best known speed-up result from $o(\log \log n)$ to $O(\log^* n)$ implied by the result of [Chang, Pettie FOCS 2017] in the \local model.
   
   Finally, we show that for every fixed constant $c \geq 2$, the deterministic \volume complexity of $c$-coloring a bounded degree tree is $\Theta(n)$, where the \volume model is a close relative of the \lca model that was recently introduced by [Rosenbaum, Suomela PODC 2020].
\end{abstract}

\section{Introduction}
For many problems in the area of big data processing it is on one hand prohibitively expensive to read the whole input, while on the other hand one is only interested in a small portion of the output at a time. One such problem comes up in the context of social networks. Many social networks suggest users different users they might want to follow or be friends with. To make such a recommendation, it is computationally too expensive to process the complete social network. Instead, it is much more efficient to look at a small local neighborhood around a user, i.e., the friends of her friends or the followers of her followers. \\
Motivated by problems where one is only interested in small parts of the output at once, Rubinfeld et al. \cite{rubinfeld2011fast} and Alon et al. \cite{alon2012space} introduced the Local Computation Model (\lca). An \lca algorithm provides query access to a fixed solution for a computational problem. That is, instead of outputting the entire solution at once, a user can ask the algorithm about specific bits of the output. To answer these queries, the \lca algorithm has probe access to the input and the main complexity measure of an \lca algorithm is the number of input probes the algorithm needs to perform to answer a given query. The specific output queries a user can ask and the specific input probes the \lca algorithm can use to learn about the input depend on the specific scenario. In the context of graph problems, the output query usually asks about the output of a given node or edge and the \lca algorithm can usually perform probes to the input of the form ``What is the $j$-th neighbor of the $i$-th node?'', where each node has a unique ID from the set $[n]$. The answer to such a probe is the ID of the specific node together with additional local information associated with that node such as for example its degree. The most well-studied type of \lca algorithms are so-called stateless \lca algorithms. The only shared state between different queries of stateless \lca algorithms is a seed of random bits. In particular, this implies that the output of an \lca algorithm is independent of the order of the queries asked by the user. In this paper, by an \lca algorithm we always mean a stateless \lca algorithm.

\paragraph{Connections to the \local model}
The \lca model is closely related to the \local model. In particular, Parnas and Ron observed that an $O(T(n))$-round \local algorithm implies an \lca algorithm with a probe complexity of $\Delta^{O(T(n))}$ where $\Delta$ denotes the maximum degree of the input graph. The reason is that one can first learn the $O(T(n))$-hop neighborhood around a node with $\Delta^{O(T(n))}$ many probes and then simulate the \local algorithm around that node to determine its output. Motivated by this connection, a recurrent theme in the study of \lca algorithms is the question in which cases it is possible to go below this straightforward simulation result. A sample of such results include a deterministic $(\Delta+1)$-coloring \lca algorithm on constant degree graphs with a probe complexity of $O(\log^* n)$ \cite{even2014deterministic}, a randomized \lca algorithm for Maximal Independent Set with a probe complexity of $\Delta^{O(\log\log \Delta)} \log n$ \cite{ghaffari2019MIS} and computing an expected $O(\log s)$-approximation for Set Cover with a probe complexity of $(st)^{O(\log s \log \log t)}$ where $s$ denotes the maximum set size and $t$ denotes the maximum number of sets a given element can be contained in \cite{grunau2020improved}. 

\paragraph{Locally Checkable Labelings}
A very fruitful research direction was and is up to this date the study of the \local complexity landscape of so-called locally checkable labeling problems (LCLs) on constant degree graphs. Informally speaking, an LCL problem requires each vertex/edge to output one of constantly many symbols such that the output of all nodes in the local neighborhood around each node satisfies some constraints imposed by the LCL. 

While the \local complexity landscape for LCLs on constant degree graphs is not completely characterized up to this date, a lot of progress has already been made. In particular, each LCL problem belongs to one of the following four classes. \cref{fig:volume_big_picture} is a graphical illustration of these four classes. 
\begin{enumerate}[label=\Alph*]
    \item Problems with a complexity of $O(1)$.
    \item Basic symmetry-breaking tasks such as $(\Delta + 1)$-coloring with a complexity between $\Omega(\log \log^* n )$ and $O(\log^* n)$.
    \item Shattering problems such as $\Delta$-coloring with a randomized complexity of $\poly(\log\log n)$ and a deterministic complexity of $\poly(\log n)$.
    \item Global problems with a complexity of $\Omega(\log n)$.
\end{enumerate}

Furthermore, it is known that the complexity regimes corresponding to $B$ and $D$ are in some sense dense and that the multiplicative gap between the randomized and deterministic \local complexity of any LCL problem is at most polylogarithmic in $n$. 

\paragraph{Lovász Local Lemma}
Arguably, one of the most useful tools in distributed computing is the distributed Lovász Local Lemma (LLL).

The LLL states that a collection of bad events can simultaneously be avoided given that each bad event only happens with a small probability and each bad event only depends on a small number of other bad events. The Distributed LLL asks an algorithm to find a point in the underlying probability space in a distributed manner that avoids all the bad events.
Formally, each bad event corresponds to a vertex in the input graph and depends on some independent random variables. Two nodes are connected by an edge if the two corresponding bad events depend on a common random variable. The output of each node is a value for each random variable the corresponding bad event depends on such that none of the bad events occurs. 

It was proven by Chang and Pettie  \cite{chang2017time} that any randomized \local algorithm with a complexity of $o(\log n)$ can be sped up to run in time $T_{LLL}(n)$ which denotes the randomized \local complexity of the distributed LLL under any polynomial criterion (\cref{def:dist_lll}) on constant degree graphs. Currently, the best upper bound to solve the distributed LLL under any polynomial criterion on constant degree graphs is $\poly \log \log n$ randomized and $\poly\log n$ deterministic \cite{moser2010constructiveLLL, fischer2017sublogarithmic, RozhonG19, ghaffari2021improved}. 
This is polynomially larger than the currently best known lower bound of $\Omega(\log \log n)$ for randomized and $\Omega(\log n)$ for deterministic algorithms \cite{brandt_LLL_lower_bound,marks2013determinacy,chang2016exp_separation}. It is a major open problem in the area of \local algorithms to close this gap. 
Finally, the works of \cite{fischer2017sublogarithmic,chang2016exp_separation,chang2017time} show that any $o(\log \log n)$ round randomized algorithm or $o(\log n)$ round deterministic algorithm can be solved deterministically in $O(\log^*n)$ rounds.
The above works and results are part of the project of classification of LCLs in the \local model, and they imply that the class (C) of problems contains exactly those that do not belong to the classes (A) and (B) and can be solved by transforming the problem instance into an LLL instance with polynomial (but not exponential) criterion and then solving the LLL instance.
The precise transformation is given by \cite[Theorem 4]{chang2017time} and it is straightforward to check that it ensures that \emph{also in the LCA model} the complexity of any problem in class (C) is asymptotically at most the complexity of the LLL under any polynomial criterion.\footnote{By construction, the nodes of the created LLL instance are constant-radius neighborhoods of the nodes of the original input instance, and essentially there is an edge between two neighborhoods if they intersect. Hence, an LCA algorithm solving the LLL instance can be simulated on the original instance by probing, for each probe in the LLL instance, also its entire neighborhood up to some constant distance, which incurs only a constant-factor overhead. Note that we use here that the input instances have constant degree. Also note that all of this holds also in the \volume model.}

\paragraph{Our Work in \lca}
Our work contributes to the understanding of the fundamental problem of characterizing LCLs in the \lca model. 

While the classes (A) and (B) of LCL problems coincide for the \local and \lca model by \cite{PARNAS2007183, even2014deterministic, chang2016exp_separation}, and for the class (D) of global problems we have some preliminary results by \cite{RosenbaumSuomela2020volume_model}, it seems that essentially nothing is known for the class (C) of shattering problems. 

As our main contribution we show that the randomized \lca complexity of the distributed LLL is $\Theta(\log n)$. Before our work, only the trivial lower bound of $\Omega(\log\log n)$ coming from the \local model \cite{brandt_LLL_lower_bound} and only the trivial upper bound $2^{\poly\log\log n}$ that follows from the Parnas-Ron reduction and the \local algorithm of Fischer and Ghaffari \cite{fischer2017sublogarithmic} were known. 
Independently of our work, Dorobisz and Kozik \cite{dorobisz2021lca_lll_hypergraphs} achieved a polylogarithmic query complexity for the problem of hypergraph coloring. This is similar to our work as their problem can be formulated as an instance of LLL (with a less restrictive LLL criterion). 

\begin{theorem}
\label{thm:lll}
The randomized \lca complexity of the distributed LLL on constant degree graphs is $\Theta(\log n)$. 
The upper bound holds for the polynomial criterion $p \le (e\Delta)^{-c}$ for some $c = O(1)$, while the lower bound holds even for the exponential criterion $p \le 2^{-\Delta}$. 
\end{theorem}
In fact, for the minimally more restrictive criterion $p < 2^{-\Delta}$, the distributed LLL can already be solved in $O(\log^* n)$ rounds in the \local model \cite{brandt_maus_uitto2019tightLLL,brandt_grunau_rozhon2020tightLLL}, which implies also a probe complexity of $O(\log^* n)$ in the \lca model~\cite{even2014deterministic}.  

Second, we show the following general speedup theorem. 

\begin{restatable}{theorem}{speedup}
\label{thm:speedup}
For any LCL $\Pi$, if there is a randomized \lca algorithm that solves $\Pi$ and has a probe complexity of $o(\sqrt{\log n})$, then there is also a deterministic \lca algorithm for $\Pi$ with a probe complexity of $O(\log^* n)$. 
\end{restatable}

Putting the two theorems together, we get that the randomized \lca complexity of all problems in class (C) is in $\Omega(\sqrt{\log n})$ and $O(\log n)$. 
This almost settles the randomized complexity of all problems in class (C). 
We conjecture that the square root is not tight and the following is, in fact, the case. 

\begin{conjecture}
Any randomized \lca algorithm that solves an LCL $\Pi$ with probe complexity $o(\log n)$ can be turned into a deterministic \lca algorithm that solves $\Pi $ and has a probe complexity of $O(\log^* n)$. 
\end{conjecture}

\begin{figure}
    \centering
    \includegraphics[width = .435\textwidth]{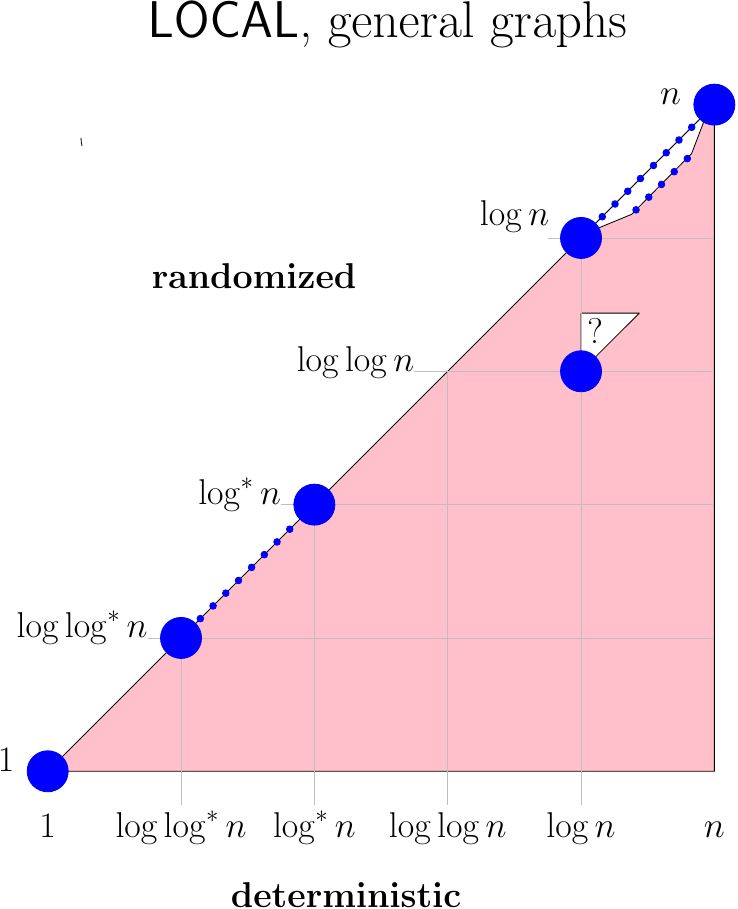}
    \includegraphics[width = .545\textwidth]{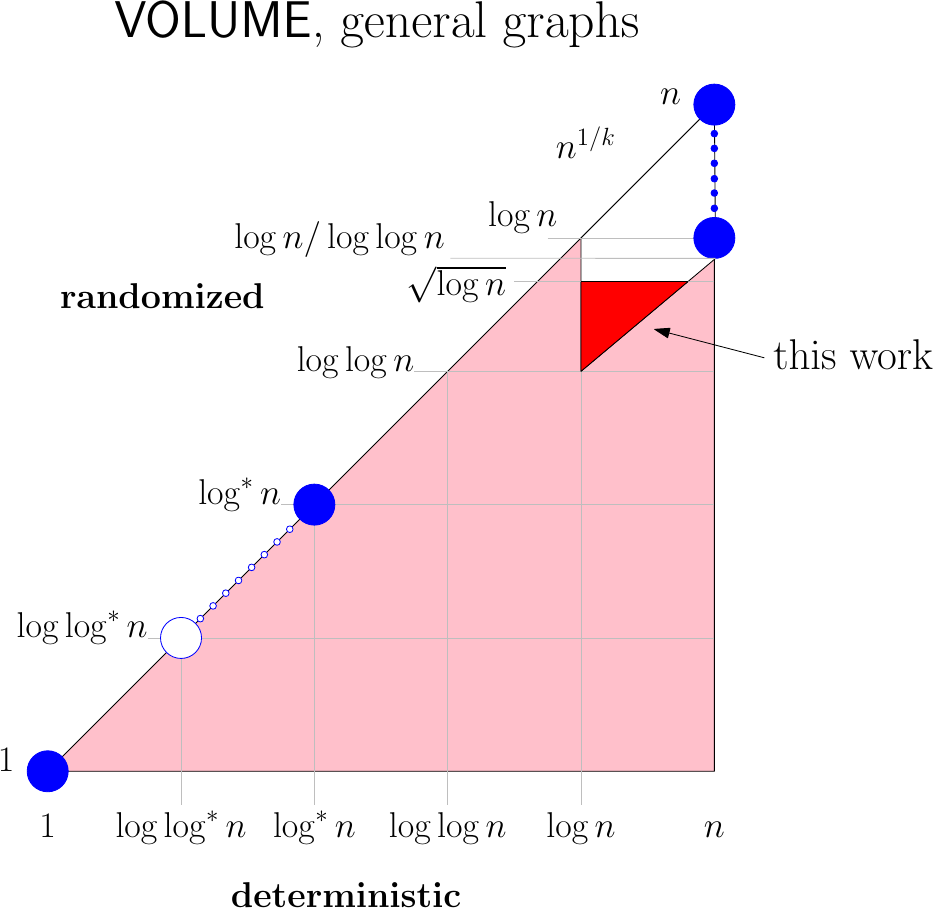}
    \caption{On the left we see the landscape of LCLs in the \local model that consists of four types of problems. 
    Note that we are especially interested in class (C) of problems that can be solved by reducing them to the distributed LLL. Their \local complexity is conjectured to be $\Theta(\log\log n)$ randomized and $\Theta(\log n)$ deterministic, but the currently known best upper bounds are only $\poly \log \log n$ randomized and $\poly \log n$ deterministic. \\ On the right we see the landscape of LCLs in the \volume model. Our new \cref{thm:speedup} implies that there are no LCLs in the bright red area. The randomized complexity of LLL in \lca is moreover settled to $\Theta(\log n)$.  }
    \label{fig:volume_big_picture}
\end{figure}

\paragraph{\volume model}
Motivated by the success story of studying the \local complexity landscape of LCLs on bounded degree graphs, Rosenbaum and Suomela \cite{RosenbaumSuomela2020volume_model} initiated the study of the complexity landscape of LCLs on constant degree graphs in the so-called \volume model. The \volume model is a close relative of the \lca model. The main difference between the two models is that the \volume model does not allow so-called far probes, that is, a \volume model algorithm can only probe a connected region around the queried vertex. Also, each node has a unique identifier in the range $\{1,2,\ldots, \poly(n)\}$ instead of $[n]$ and the nodes only have access to private randomness instead of shared randomness. The \volume model is a bit cleaner to work with than the \lca model. Nevertheless, there are generic simulation results that often allow to transfer upper and lower bounds between the \lca and the \volume model. Mainly by transferring known results and techniques from the \local model to the \volume model, Rosenbaum and Suomela obtained the result illustrated in \cref{fig:volume_big_picture}---the first rough outline of the complexity landscape of LCLs on constant degree graphs in the \volume model.

\paragraph{Our Work in the \volume model}

We show the following theorem about coloring bounded degree trees in the \volume model.

\begin{restatable}{theorem}{coloring} 

\label{theorem:constant_coloring_theorem}
Let $c \geq 2$ be arbitrary. Then, the deterministic \volume complexity of $c$-coloring bounded degree trees with maximum degree $\poly(c)$ is $\Theta(n)$.
\end{restatable}
We note that the deterministic and randomized \local complexity for the same problem is $\Theta(\log n)$ assuming $c \geq 3$.
It would be very interesting to extend the result to the \lca model. The main reason we did not manage to extend it to the \lca model are far probes. These are especially difficult to handle if we want the identifiers from $[n]$. If we do not allow the \lca algorithm to perform far probes, but we want identifiers from $[n]$ instead of $\poly(n)$, then an approach similar to ours gives an $\Omega(n^{1-o(1)})$ lower bound. Our argument also breaks down for randomized algorithms. Hence, it is an interesting problem to prove any randomized polynomial lower bound or to come up with an efficient randomized algorithm.

\paragraph{Further Related Work: Connection of \lca to Parallel Algorithms}
As the only shared state between queries of \lca algorithms is the random seed, after distributing the random seed to all processors, the processors can answer queries independent of each other and therefore in parallel. Moreover, many randomized \lca algorithms also work with $k$-wise independent random bits for $k = O(\poly(\log n))$ and hence standard techniques often allow for random seeds of polylogarithmic length \cite{alon2012space}. \\
There also exists a close connection between the Massively Parallel Computation (\mpc) model \cite{karloff2010mpc}---a theoretical model to study map-reduce type algorithms---and the \lca model. In the \mpc model, the input of a computational problem is distributed across multiple machines. The computation proceeds in synchronous round. At the beginning of each round, each machine can perform some local computation. Afterwards, each machine can send and receive messages to other machines with the restriction that each machine can send and receive at most as many bits as the size of its local memory. The most stringent regime in the \mpc model is the sublinear memory regime where each machine has a local memory of $O(n^{\alpha})$ bits for some $\alpha \in (0,1)$. In this regime, a machine cannot get a global view of the graph. Hence, most known algorithms gather the local neighborhood around each node followed by executing a local algorithm on the neighborhood to deduce the output of a node. However, the local neighborhood around a node often contains too many vertices to fit into a single machine. In such a scenario, it is important that we can compute the output of a node by only considering a minuscule fraction of its local neighborhood. This is exactly what \lca algorithms try to accomplish. In particular, Ghaffari and Uitto \cite{ghaffari2019MIS} devised a sparsification technique that resulted in state-of-the-art algorithms for Maximal Independent Set both in the \lca model and the \mpc model. 
Moreover, \cite{RosenbaumSuomela2020volume_model} stated a generic result that directly allows to transfer results from the \lca model to the \mpc model.

\section{Our Method in a Nutshell}

In this section we informally present the high-level ideas of our proofs. 

\paragraph{The Gap Result of \cref{thm:speedup}}

We start to describe the high-level idea of the proof that there is no LCL with a randomized/deterministic \lca /\volume complexity between $\omega(\log^* n)$ and $o(\sqrt{\log n})$. 
This is the conceptually simplest result and it works by directly adapting ideas known from the \local model. 
The main trick is to consider the deterministic \volume model where the identifiers can come from an exponential instead of a polynomial range (cf. Section 7.1 in \cite{RosenbaumSuomela2020volume_model}). 
The result then follows from two observations. First, a variant of the Chang-Pettie speedup \cite{chang2017time} shows that any \volume algorithm with a probe complexity of $o(n)$ that works with exponential IDs can be sped up to have a probe complexity of $\Theta(\log^*n)$. 
Second, a variant of the Chang-Kopelowitz-Pettie derandomization \cite{chang2016exp_separation} shows that any randomized algorithm with probe complexity $o(\sqrt{\log n})$ can be derandomized to give a deterministic $o(n)$-probe algorithm that works with exponential IDs. The root comes from the exponential IDs: during the argument we apply a union bound over all bounded degree $n$-node graphs equipped with unique IDs from an exponential range.

\paragraph{The LLL Complexity of \cref{thm:lll}}

The upper bound can directly be proven by adapting the LLL algorithm of \cite{fischer2017sublogarithmic} for the \local model to the \volume  model. To prove an $\Omega(\log n)$ lower bound, we prove a corresponding lower bound of $\Omega(\log n)$ for the Sinkless Orientation Problem, which can be seen as an instance of the distributed LLL (\cref{def:sinkless_orientation}). To prove this lower bound, we use the same high-level proof idea as described in the previous paragraph. However, to get a tight $\Omega(\log n)$ lower bound we need additional ideas. More concretely, the number $\sqrt{\log n}$ is an artifact of a union bound over  $2^{O(n^2)}$ many non-isomorphic ID-labeled graphs. As we show the Sinkless Orientation lower bound on bounded-degree trees, it is actually sufficient to union bound over all non-isomorphic ID-labeled bounded-degree trees. As the number of non-isomorphic unlabeled trees is $2^{O(n)}$, we already made progress. However, this itself is still not enough as we need identifiers from an exponential range to speed the algorithm up to $\Theta(\log^* n)$ and there are $2^{O(n^2)}$ many ways to assign unique  exponential IDs. Note that even if we would assign IDs from a polynomial range, there still would be $2^{O(n\log n)}$ ways to do it, which would only allow to hope for an $\Omega(\log n/\log \log n)$ bound.

To circumvent this issue we borrow an idea from a parallel paper~\cite{brandt_chang_grebik_grunau_rozhon_vidnyaszky2021LCLs_on_trees_descriptive} where the technique of ID graphs is developed to overcome a different issue. 
The idea is as follows: we will work in the deterministic model with exponential IDs. 
However, we promise that the ID assignment satisfies certain additional properties. More concretely, we construct a so-called ID graph (which is not to be confused with the actual input graph). Each node in the ID graph corresponds to one of the exponentially many IDs. Moreover, the maximum degree of the ID graph is constant and two nodes in the input graph that are neighbors can only be assigned IDs that are neighbors in the ID graph.  Having restricted the ID assignment in this way, it turns out that we only need to union bound over $2^{O(n)}$ different ID-labeled trees. Hence, a $o(\log n)$ randomized \volume algorithm would imply a $o(n)$ deterministic \volume algorithm that works on graphs with exponential IDs that satisfy the constraints imposed by the ID graph. Unfortunately, due to this additional restriction, we cannot simply speed-up the algorithm to the $\Theta(\log^* n)$ complexity. 
However, it turns out that the famous round elimination lower bound for Sinkless Orientation from \cite{brandt_LLL_lower_bound} works even relative to an ID graph (the formal proof is in fact simpler as one does not pass through a randomized model). 
This finishes the lower bound proof for the \volume model. This directly leads to the same lower bound for the \lca model by a result of \cite{goos_nonlocal_probes_16}.

\paragraph{The Deterministic \volume Complexity of Coloring Bounded Degree Trees with Constantly Many Colors is $\Theta(n)$ (\cref{theorem:constant_coloring_theorem})}

We first sketch a well-known proof of the $\Omega(\log n)$ \local lower bound for the same problem. We then show how to adapt this proof to the \volume model. To prove a \local lower bound of $\Omega(\log n)$, one fools the algorithm by running it on a graph having a girth of $\Omega(\log n)$ and a large constant chromatic number instead of a tree. Due to the high girth, the graph looks in the $o(\log n)$-hop neighborhood around each node like a tree. Hence, any \local algorithm with a round complexity of $o(\log n)$ cannot detect that it is not run on a valid input. Moreover, as the chromatic number of this high-girth graph is strictly larger than the constantly many colors available to the algorithm, there need to exist two neighboring nodes that get assigned the same color by the \local algorithm. To arrive at a contradiction, one can now construct a tree that contains these two neighboring nodes and moreover the \local algorithm will again assign these two neighboring nodes the same color. The main difficulty in transferring this proof to the \volume model is that $o(n)$ probes might suffice to find a cycle in the high-girth graph. Hence, the \volume algorithm might detect that the input is not a tree. To make it harder for the \volume algorithm to find a cycle, we add additional vertices and edges to this high-girth graph without introducing any new cycles. In fact, the resulting graph will have infinitely many vertices (though one could restrict oneself to work with a finite graph). As we still want to give the deterministic algorithm with a probe complexity of $o(n)$ the illusion that the graph only contains $n$ many vertices, we assign each node in the graph an identifier from $[n^{10}]$. This identifier can of course not be unique. In order to prevent the \volume algorithm from detecting duplicate identifiers, we assign each node an identifier uniformly and independently at random. Assigning identifiers in that way, it is unlikely for the algorithm to find a duplicate ID. Moreover, the random identifiers do not provide any information about the topology of the graph and thus one can show that it is unlikely that the algorithm finds a cycle. By a probabilistic method argument, this allows us to arrive at a contradiction in a similar manner as in the \local model lower bound.
\subsection{Definitions}

\paragraph{Notation}

We use the classical graph-theoretical notation, e.g., we write $G=(V,E)$ for an unoriented graph. 
A half-edge is a pair $h = (v,e)$, where $v \in V$, and $e \in E$ is an edge incident to $v$. 
Often we assume that $G$ additionally carries a labeling of vertices or half-edges.
We use $B_G(u,r)$ to denote the ball of radius $r$ around a node $u$ in $G$. 
When talking about half-edges in $B_G(u,r)$, we talk about all half-edges $(v,e)$ such that $v \in B_G(u, r)$. For example, $B_G(u, 0)$ contains all half-edges incident to $u$. 

\begin{definition}[LCLs]
An LCL problem (or simply LCL) $\Pi$ for a constant degree graph is a quadruple $(\Sigma_{in}, \Sigma_{out}, r, \fP)$ where $\Sigma_{in}$ and $\Sigma_{out}$ are finite sets, $r$ is a positive integer, and $\fP$ is a finite collection of $\Sigma_{in}$-$\Sigma_{out}$-labeled graphs.
A correct solution for an LCL problem $\Pi$ on a $\Sigma_{in}$-labeled graph $(G, f_{in})$ is given by a half-edge labeling $f_{out} \colon H(G) \to \Sigma_{out}$ such that, for every node $v \in V(G)$, the triple $(B_G(v,r), f'_{in}, f'_{out})$ is isomorphic to a member of $\fP$, where $f'_{in}$ and $f'_{out}$ are the restriction of $f_{in}$ and $f_{out}$, respectively, to $B_G(v,r)$.
\end{definition}

Intuitively, the collection $\fP$ provides the constraints of the problem by specifying how a correct output looks \emph{locally}, depending on the respective local input.
From the definition of a correct solution for an LCL problem it follows that members of $\fP$ that have radius $> r$ can be ignored.

\begin{definition}[\lca model \cite{rubinfeld2011fast} \cite{alon2012space}]
\label{def:lca}
In the \lca model, each node is assigned a unique ID from the set $[n]$. Moreover, each node is equipped with a port numbering of its edges and each node might have an additional input labeling. The \lca algorithm needs to answer queries. That is, given a vertex/edge, it needs to output the local solution of the vertex/edge in such a way that combining the answers of all vertices/edges constitutes a valid solution. To answer a query,  the algorithm can probe the input graph. A probe consists of an integer $i \in [n]$ and a port number and the answer to the probe is the local information associated with the other endpoint of the edge corresponding to the specific port number of the vertex with ID $i$. The answer to a query is only allowed to depend on the input graph itself and possibly a shared random bit string in case of randomized \lca algorithms. The complexity of an \lca algorithm is defined as the maximum number  of probes the algorithm needs to perform to answer a given query, where the maximum is taken over all input graphs and all query vertices/edges. A randomized \lca algorithm needs to produce a valid complete output (obtained by answering the query for each vertex) with probability $1 - 1/n^c$ for any desirably large constant $c$.
\end{definition}
\begin{definition}[\volume model \cite{RosenbaumSuomela2020volume_model}]
The \volume model is very similar to the \lca model and hence we only discuss the differences. The IDs in the \volume model are from the set $\{1,2,\ldots,\poly(n)\}$, as in the \local model, instead of $[n]$. Moreover, a \volume model algorithm is confined to probe a connected region. In the case of randomized \volume algorithms, each node has a private source of random bits which is considered as part of the local information and is therefore returned together with the ID of a given vertex. We note that the shared randomness of the \lca model is strictly stronger than the private randomness of the \volume model and therefore the \lca model is strictly more powerful than the \volume model. 
\label{def:volume}

\end{definition}

\begin{definition}[\local model \cite{linial92}, \cite{peleg00}]
 
\label{def:local}
In the \local model of distributed computing, the goal is to compute a graph problem in a network. Each node in the network corresponds to a computational entity and is equipped with a unique identifier in $\{1,2,\ldots,\poly(n)\}$. In the beginning, each node in the network only knows its ID, some global parameters like $n$ and the maximum degree $\Delta$ of the network, and perhaps some local input. Computation proceeds in synchronous rounds. That is, in each round, each processor can first perform unbounded local computation and then send each of its neighbors a message of unbounded length. Each node has to decide at some point that it terminates and then it must output its local part of the global solution to the given problem. The round complexity of a distributed algorithm is the number of rounds until the last round terminates.
\end{definition}

\begin{definition}[Sinkless Orientation]
\label{def:sinkless_orientation}
The Sinkless Orientation problem asks to orient each edge of a given input graph in such a way that each vertex of sufficiently high constant degree is incident to at least one outgoing edge.
\end{definition}

\begin{lemma}[Lovász Local Lemma (LLL), \cite{lovasz74}]
\label{lem:LLL}
We denote with $\{X_1, X_2\ldots,X_m\}$ a set of mutually independent random variables and with $\mathcal{E}_1, \ldots, \mathcal{E}_n$ probabilistic events. Each $\mathcal{E}_i$ is a function of some subset of the random variables $X_1, X_2,\ldots, X_m$ and this subset is denoted by vbl($\mathcal{E}_i$). We say that $\mathcal{E}_i$ and $\mathcal{E}_j$ depend on a common random variable if $vbl(\mathcal{E}_i) \cap vbl(\mathcal{E}_j) \neq \emptyset$. Assume that there is some $p < 1$ such that for each $1 \leq i \leq n$, we have $P(\mathcal{E}_i) \leq p$, and let $d$ be a positive integer such that each $\mathcal{E}_i$ shares a random variable with at most $d$ other $\mathcal{E}_j$, $j \neq i$. If $4pd \leq 1$, then there exists an assignment of values to the random variables such that none of the events $\mathcal{E}_i$ occurs. 
\end{lemma}

\begin{definition}[Distributed Lovász Local Lemma]
\label{def:dist_lll}
The constructive LLL asks to find a concrete assignment of the random variables $X_1,X_2,\ldots,X_m$ such that none of the bad events $\mathcal{E}_1, \mathcal{E}_2, \ldots, \mathcal{E}_m$ occurs. In the Distributed LLL, the set of nodes of the input graph is simply the set of bad events $\{\mathcal{E}_1, \mathcal{E}_2, \ldots, \mathcal{E}_m\}$. Moreover, $\mathcal{E}_i$ and $\mathcal{E}_j$ are connected by an edge iff $vbl(\mathcal{E}_i) \cap vbl(\mathcal{E}_j) \neq \emptyset$ and $i \neq j$. In the end, each node $\mathcal{E}_i$ needs to know the assignment of values to all the random variables in $vbl(\mathcal{E}_i)$. These assignments need to be consistent and they need to simultaneously avoid all the bad events. Often, one considers variants of the LLL that further restrict the space of allowed input instances by replacing the criterion $4pd \leq 1$ with a more restrictive inequality. A polynomial criterion is one of the form $pf(d)\leq 1$, where $f(d)$ is some polynomial in $d$. An exponential criterion is one of the form $pf(d) \leq 1$, where $f(d)$ is exponential in $d$.

\end{definition}

By directing each edge independently with probability $1/2$ in each direction, one can view Sinkless Orientation as an instance of the Distributed LLL that satisfies the exponential criterion $p2^{d} \leq 1$. In particular, this implies that the $\Omega(\log n)$ Sinkless Orientation lower bound directly implies an $\Omega(\log n)$ lower bound for the LLL under the exponential criterion $p2^{d} \leq 1$.

\section{Preliminaries}

The following is a well-known fact that follows from a straightforward simulation idea.

\begin{lemma}[Parnas-Ron reduction, \cite{PARNAS2007183}]
\label{lem:parnas_ron}
Any \local algorithm with a round complexity of $t(n)$ can be converted into an \lca /\volume algorithm with a probe complexity of $\Delta^{O(t(n))}$, where $\Delta$ denotes the maximum degree of the input graph.
\end{lemma}

The following lemma is a restatement of Theorem 3 in \cite{goos_nonlocal_probes_16}.

\begin{lemma}
\label{lem:rand_far_probes}
Suppose that there is a randomized \lca algorithm with probe complexity $t(n)$ that solves an LCL $\Pi$. Then there is a randomized \lca algorithm with probe complexity $t'(n) = t(\poly(n))$ that solves $\Pi$, does not perform any far probes, and works even if the unique identifiers come from a polynomial range (instead of from $[n]$). This also holds if we restrict ourselves to trees. 
\end{lemma}

We use this lemma to extend a lower bound that we prove for the \volume model to \lca algorithms (that are allowed to perform far probes).

We also need to use the following lemma which informally states that far probes are of no use for deterministic \lca algorithms with a small probe complexity.
\begin{theorem}[Theorem 1 in \cite{goos_nonlocal_probes_16}]
\label{thm:goos}
Any LCL problem that can be solved in the \lca model deterministically with probe complexity $t(n)$ can be solved with a deterministic round complexity of $t'(n) = t(n^{\log n})$ in the \local model
provided $t(n) = o(\sqrt{\log n})$. 
\end{theorem}
For our purposes, we actually need a slightly stronger version of \cref{thm:goos} which is, in fact, what the proof of \cref{thm:goos} in \cite{goos_nonlocal_probes_16}  gives.

\begin{theorem}
\label{thm:goos_refined}
Any LCL problem that can be solved in the \lca model deterministically with probe complexity $t(n)$ can be solved deterministically with a probe complexity of $t'(n) = t(n^{\log n})$ in the \volume model provided $t(n) = o(\sqrt{\log n})$.
\end{theorem}

\section{Warm-up: Speedup of Randomized Algorithms in \lca}
\label{subsec:volume_speedup}

In this section we prove \cref{thm:speedup} that we restate here for convenience. 

\speedup*

In fact, we show that the resulting deterministic \lca algorithm with a probe complexity of $O(\log^* n)$ can also be turned into a \volume algorithm with the same probe complexity.
The following lemma is proven with a similar argument as Theorem $3$ in \cite{chang2016exp_separation}. 

\begin{lemma}[Derandomization in \lca]
\label{lem:volume_derandomization}
If there exists a randomized \lca algorithm $\fA$ with a probe complexity of $t(n) = o(\sqrt{\log n})$ for a given LCL $\Pi$, then there also exists a deterministic \volume algorithm $\fA'$ for $\Pi$ with a probe complexity of $o(n)$ and  where the identifiers are from $[2^{O(n)}]$. 
\end{lemma}
\begin{proof}
Let $r$ be the local checkability radius of $\Pi$. First, we apply \cref{lem:rand_far_probes} to convert the algorithm $\fA$ into one having the same asymptotic complexity but which does not use any far probes and assumes only identifiers from a polynomial range. 
Without loss of generality, we can assume that $\mathcal{A}$ works in a setting where, instead of being assigned an identifier, each node has access to a private random bit string (in other words, $\mathcal A$ works in the \volume model with the addition of shared randomness). This is justified as, with access to private randomness in each node, the algorithm can generate unique identifiers with probability $1 - 1/\poly(n)$, as the first $O(\log n)$ random bits of each node are unique with probability $1 - 1/\poly(n)$. 

Next, let $\fG_n$ denote the set of all $n$-node graphs (up to isomorphism) of maximum degree $\Delta$ with each vertex labeled with a unique identifier from $[2^{O(n)}]$ and where each vertex has an input label from the finite set of input labels from the given LCL. The number of unlabeled graphs of maximum degree at most $\Delta$ is $2^{O(n \log n)}$ as it can be described by $O(n \cdot \Delta \cdot \log n)$ bits of information. 
Next, for a fixed $n$-node graph, the number of its labelings with identifiers from $[2^{O(n)}]$ is upper bounded by $(2^{O(n)})^n = 2^{O(n^2)}$ and the number of distinct input label assignments is upper bounded by $2^{O(n)}$.
Hence, the number of (labeled) graphs in $\fG_n$ is strictly smaller than some suitably chosen $N = 2^{O(n \log n)} \cdot 2^{O(n^2)} \cdot 2^{O(n)} = 2^{O(n^2)}$.

Now, let $\rho : [2^{O(n)}] \rightarrow \{0,1\}^\mathbb{N}$ be a function that maps each ID to a stream of bits, chosen uniformly at random from the space of all such functions.
Similarly, let $\rho^*$ be a bit string chosen uniformly at random from $\{0,1\}^\mathbb{N}$.
Consider the algorithm $\mathcal{A}_{\rho, \rho^*}$ solving $\Pi$ on all graphs $G \in \fG_n$ that is defined as follows.
First, $\mathcal{A}_{\rho, \rho^*}$ (internally) maps each identifier $I_v$ it sees (at some node $v$) in $G$ to a bit string by applying the function $\rho$, and then it simulates algorithm $\mathcal{A}$, where the private random bit string that $\mathcal A$ has access to in each node $v$ is given by $\rho(I_v)$, the shared random bit string is given by $\rho^*$, and the input parameter given to $\mathcal A$ describing the number of nodes is set to $N$.
Equivalently, this can be seen as running $\mathcal A$ (with randomness provided by $\rho$ and $\rho^*$) on the graph $H$ obtained from $G$ by adding $N - n$ isolated nodes (where we are only interested in the output of $\mathcal A$ on the nodes of $H$ that correspond to nodes in $G$).

Since, on $N$-node graphs, $\mathcal A$ provides a correct output with probability at least $1 - 1/N$, and $\rho$ and $\rho^*$ are chosen uniformly at random, we see that, for every $G \in \fG_n$, the probability that $\mathcal{A}_{\rho, \rho^*}$ fails on $G$ is at most $1/N$.
Since the number of graphs in $\fG_n$ is strictly smaller than $N$, it follows that there are a function $\rho_{\dett} : [2^{O(n)}] \rightarrow \{0,1\}^\mathbb{N}$ and a bit string $\rho^*_{\dett} \in \{0,1\}^\mathbb{N}$ such that $\mathcal{A}_{\rho_{\dett}, \rho^*_{\dett}}$ does not fail on any graph in $\fG_n$.

As the runtime of $\mathcal{A}_{\rho_{\dett}, \rho^*_{\dett}}$ is equal to the runtime of $\mathcal A$ on $N$-node graphs, we can conlude that $\mathcal{A}_{\rho_{\dett}, \rho^*_{\dett}}$ is a deterministic \volume algorithm with probe complexity $t(N) = o(\sqrt{\log N}) = o(\sqrt{\log 2^{O(n^2)}}) = o(n)$, as desired.
\end{proof}

Similarly, the next lemma is a variant of Theorem 6 in \cite{chang2016exp_separation} and the discussion in Section 7.1 in \cite{RosenbaumSuomela2020volume_model}. 
\begin{lemma}[Speedup in \volume with exponential identifiers]
\label{lem:exponential_id_speedup}
If there exists a deterministic \volume algorithm $\fA$ for an LCL $\Pi = (\Sigma_{in}, \Sigma_{out}, r, \fP)$ with a probe complexity of $o(n)$ that works with unique identifiers from $[2^{O(n)}]$, then there also is a deterministic \volume algorithm $\fA'$ for $\Pi$ with a probe complexity of $O(\log^* n)$. 
\end{lemma}
\begin{proof}
Let $n_0$ be a big enough constant. The algorithm $\fA'$ does the following on a given input graph $G$ with maximum degree $\Delta$: first, it uses the algorithm of Even et al. \cite{even2014deterministic} to construct a coloring of the power graph $G^{n_0+r}$---the graph with the vertex set of $G$ and where two nodes are connected by an edge iff they have a distance of at most $n_0 + t$ in the graph $G$---with $\Delta^{n_0+r} + 1 = 2^{O(n_0)}$ colors with a probe complexity of $O(\log^* n)$. Then we interpret those colors as identifiers and run the algorithm $\fA$ on $G$ but we tell the algorithm that the input graph has $n_0$ nodes instead of $n$. The query complexity of $\fA'$ is $O(\log^* n \cdot o(n_0)) = O(\log^* n)$. 

To see that $\fA'$ produces a valid output, note that if $\fA'$ fails at a node $u$, we may consider all the nodes that were needed for the simulation of $\fA$ in the $r$-hop neighborhood of $u$ together with their neighbors---the number of such nodes is bounded by $\Delta \cdot \Delta^r\cdot o(n_0) < n_0$, by choosing $n_0$ large enough. But as all those nodes are labeled by unique identifiers from $[2^{O(n_0)}]$, we can get a graph of size $n_0$ on which the original algorithm $\fA$ fails, a contradiction. 
\end{proof}
\cref{thm:speedup} now follows from \cref{lem:volume_derandomization,lem:exponential_id_speedup}. 

We remark that by using polynomial instead of exponential identifiers in \cref{lem:volume_derandomization}, we would get that a randomized \lca /\volume algorithm of complexity $t(n) = o(\log n /\log \log n)$ can be derandomized to a deterministic algorithm with complexity $t(2^{O(n\log n)}) = o(n)$. The term $2^{O(n\log n)}$ comes from a union bound over all $n$-node graphs of maximum degree $\Delta$ labeled with polynomial-sized unique identifiers. This is the reason for the segment between the complexity pairs $[\log n, \log\log n]$
and $[n, \log n / \log\log n]$ in \cref{fig:volume_big_picture} (cf. Section 1.2 and Figure 2 in \cite{RosenbaumSuomela2020volume_model} that does not differentiate between $\Theta(\log n)$ and $\Theta(\log n / \log \log n)$).

\section{The $\Omega(\log n)$ \volume Lower Bound for Sinkless Orientation}
\label{subsec:id_graph}

In this section, we prove the lower bound of \cref{thm:lll}. We prove it by providing the respective lower bound for the Sinkless Orientation problem (\cref{def:sinkless_orientation}) on trees.

\begin{theorem}
\label{thm:so_lower}
There is no randomized LCA algorithm with probe complexity $o(\log n)$ for the problem of Sinkless Orientation or $\Delta$-coloring, even if the input graph is a tree with a precomputed $\Delta$-edge coloring. 
In particular, the \lca complexity of LLL is $\Omega(\log n)$, even in the regime $p \le 2^{-\Delta}$.  
\end{theorem}

The general idea of the proof is the same as in \cref{subsec:volume_speedup}: we want to derandomize the assumed $o(\log n)$-probe randomized algorithm for Sinkless Orientation to a deterministic $o( n)$-probe \volume algorithm, this time restricting what constitutes a valid ID assignment. 
We then reduce the problem of showing that such a \volume algorithm does not exist to the problem of showing that there does not exist a nontrivial deterministic \local algorithm for Sinkless Orientation. We will explain later in more detail what we mean by nontrivial. Finally, we conclude the proof by showing that such a \local algorithm does not exist.

As we saw in \cref{subsec:volume_speedup}, a direct application of the derandomization by Chang-Kopelowitz-Pettie allows us only to argue that an $o(\sqrt{\log n})$-probe algorithm leads to an $o(n)$-probe deterministic algorithm, so we need to do better. 
There are two obstacles to obtaining a union bound over $2^{O(n)}$ graphs in \cref{lem:volume_derandomization}:
\begin{enumerate}
    \item the number of $n$-node graphs of maximum degree $\Delta$ is bounded only by $2^{O(n \log n)}$,
    \item the number of ways of labeling $n$ objects with labels from $[2^{O(n)}]$ is $2^{O(n^2)}$. 
\end{enumerate}

To get an $\Omega(\log n)$ lower bound, both terms need to be improved to $2^{O(n)}$. 

This is easy with the first term -- in fact, the whole lower bound works even if we restrict ourselves to trees. The number of trees with maximum degree $\Delta$ can easily be upper bounded by $\Delta^{O(n)}$ and an even stronger upper bound of $2.96^n$ is known.

The issue is with the second bullet point---the number of labelings of $n$ objects with labels from range $2^{O(n)}$ clearly cannot be improved from $2^{O(n^2)}$.

To decrease the number of possibilities, we need to restrict our space of labelings of a tree with unique identifiers in such a way that the number of possibilities drops to $2^{O(n)}$, yet this restriction should not make it easier to solve Sinkless Orientation in the \volume model. This is done by the ID graph technique developed in a parallel paper~\cite{brandt_chang_grebik_grunau_rozhon_vidnyaszky2021LCLs_on_trees_descriptive} for a different purpose. 
An ID graph $H$ is a graph that states which pairs of identifiers are allowed for a pair of neighboring nodes of the input tree. In our case, one should think about it as a high-girth high chromatic number graph on $2^{O(n)}$ nodes with each node representing an identifier. The  girth of the graphs is at least $\Theta(n)$ and its chromatic number at least $\Delta$ (i.e., the maximum degree of the \emph{input graph}). Our definition is a little subtler due to the fact that we work on edge-colored trees where arguments are usually easier. 

\paragraph{ID graph}
We now define the ID graph, prove that the number of labelings of $n$-node trees consistent with it is bounded by $2^{O(n)}$ in \cref{lem:distinct_trees}, and then prove \cref{lem:so_rand_to_det}.
Each vertex of the ID graph can be considered as an identifier that will later be used to provide IDs to the considered input graph.

\begin{definition}[ID graph]
\label{def:id_graph}
Let $R$ and $\Delta$ be positive integers.
An ID graph $H = H(R, \Delta)$ is a collection of graphs $H_1, H_2, \dots, H_\Delta$ such that the following hold:
\begin{enumerate}
    \item For all $i,j$ satisfying $1 \le i, j \le \Delta : V(H_i) = V(H_j)$; we use $V(H)$ to denote the set of vertices in $H$, that is, $V(H) = V(H_1)$, 
    \item $|V(H)| = \Delta^{10R}$,
    \item $\forall v \in V(H), \forall 1 \leq i \leq \Delta : 1 \le \deg_{H_i}(v) \le \Delta^{10}$,
    \item $\girth(H) \ge 10R$,
    \item Any independent set of $H_i$ has less than $|V(H)| / \Delta$ vertices. 
\end{enumerate}
\end{definition}

\begin{lemma}[ID graph existence]
\label{lem:id_graph_existence}
There exists an ID graph $H = H(R, \Delta)$ for all sufficiently large $R,\Delta > 0$. 
\end{lemma}
This lemma is proved in a parallel paper~\cite{brandt_chang_grebik_grunau_rozhon_vidnyaszky2021LCLs_on_trees_descriptive} developing the technique for a different purpose. For completeness we also leave a proof to \cref{sec:appendix}.

As there are possibly many ID graphs that satisfy the given conditions, from now on, whenever we write $H(R, \Delta)$, we mean the lexicographically smallest ID graph $H(R, \Delta)$. 

\begin{definition}[Proper $H$-labeling of a $\Delta$-edge-colored tree]
Let $T$ be a tree having a maximum degree of at most $\Delta$ and whose edges are properly colored with colors from $[\Delta]$. A proper $H$-labeling of $T$ with an ID graph $H$ is a labeling $h : V(T) \rightarrow V(H)$ of each vertex $u \in V(T)$ with a vertex $h(u) \in V(H)$ such that whenever $u,v \in V(T)$ are incident to a common edge colored with color $c \in [\Delta]$, then $h(u)$ and $h(v)$ are neighboring in $H_c$. 
\end{definition}

\begin{definition}[Solving an LCL relative to $H$]
\label{def:relative_to_H}

When we say that a deterministic \volume or \local algorithm for an LCL $\Pi$ works \emph{relative to an ID graph $H$}, we mean that the algorithm works if the unique node identifiers are replaced by a proper $H(n, \Delta)$-labeling for instances of size $n$ and with maximum degree at most $\Delta$. 

We also say that an algorithm works \emph{relative to ID graphs $\{H_n\}_{n \in \mathbb{N}}$} if it works for every $n$-node input graph that is $H$-labeled by $H_n$. 
\end{definition}

\begin{observation}
If a deterministic local algorithm $\fA$ solves a problem $\Pi$ of checkability radius $t$ in $r$ rounds (in case of the LOCAL model) or with $r$ probes (in case of the volume model) with $\Delta^{O(r)}$-sized identifiers, then it also solves $\Pi$ relative to $H(t+r, \Delta)$. 
\end{observation}
\begin{proof}
The validity of $\fA$ depends on all possible ways of labeling an $(r+t)$-hop neighborhood of a vertex $u$ with identifiers, as the correctness of the output at $u$ depends only on the outputs given by $\fA$ at each vertex $v$ in $u$'s $t$-hop neighborhood (by the definition of an LCL problem), and each such output depends only on the $r$-hop neighborhood of the respective node $v$. But the set of allowed labelings relative to $H(t+r, \Delta)$ is a subset of all labelings with unique identifiers. 
\end{proof}

\begin{lemma}
\label{lem:distinct_trees}
The number of non-isomorphic $n$-node trees with maximum degree $\Delta = O(1)$ that are labeled with a proper $\Delta$-edge coloring and an $H$-labeling for $H = H(n, \Delta)$, is $2^{O(n)}$. 
\end{lemma} 
\begin{proof}
There are $O(2.96^{n})$ non-isomorphic unrooted trees on $n$ vertices \cite{oeis_num_of_trees}. Additionally, there are at most $\Delta^{n - 1}$ ways of assigning edge colors from the set $[\Delta]$ once the tree is fixed. 
Hence, there are $2^{O(n)}$ non-isomorphic trees labeled with edge colors. 

For any such tree $T$, pick an arbitrary vertex $u$ in it. 
There are $|V(H)| = \Delta^{O(n)}$ ways of labeling $u$ with a label from $H$ (Property 2 in \cref{def:id_graph}). 
Once $u$ is labeled with a label $h(u) \in V(H)$, we can construct the labeling of the whole tree $T$ by gradually labeling it, vertex by vertex, always labeling a node whose neighbor was labeled already. Every time we label a new node $v$ with the label $h(v) \in V(H)$ such that $v$ is adjacent to an already labeled node $w$ via an edge of color $c$, we have $\poly(\Delta)$ possible choices for the label of $v$, since this is the degree of $h(w)$ in $H_c$ (Property 3 in \cref{def:id_graph}). Hence, the total number of $H$-labelings of $T$ is $2^{O(n)}$. 

Putting everything together, the number of non-isomorphic trees labeled by edge colors from $[\Delta]$ and IDs from $H$ is $2^{O(n)}$. 
\end{proof}

The following lemma is analogous to \cref{lem:volume_derandomization}, i.e., the derandomization from \cite{chang2016exp_separation}, but working relative to the ID graph and only on the set of trees allows us to derandomize a randomized \volume algorithm with probe complexity $t(n)$ to obtain a deterministic \volume algorithm with probe complexity $t(2^{O(n)})$, while the original construction only obtains a deterministic \volume algorithm with probe complexity $t(2^{O(n^2)})$.

\begin{lemma}[From randomized \lca to deterministic  \volume relative to an ID graph]
\label{lem:so_rand_to_det}
If there exists a randomized \lca algorithm with probe complexity $t(n) = o(\log n)$ for Sinkless Orientation on trees with maximum degree $\Delta = O(1)$ that are properly $\Delta$-edge colored, then there also exists a deterministic \volume algorithm for Sinkless Orientation on trees with maximum degree $\Delta = O(1)$ that are properly $\Delta$-edge colored with probe complexity $o(n)$ relative to ID graphs $\{H(n, \Delta)\}_{n \geq n_0}$ where $n_0$ is a large enough constant. 
\end{lemma}
\begin{proof}
We omit the proof as it can be proven in the exact same way as \cref{lem:volume_derandomization}, except that now we tell the randomized algorithm that the number of nodes is $N = 2^{O(n)}$ instead of $N = 2^{O(n^2)}$ as we need to union bound over a smaller number of labeled graphs. One also needs to make use of the fact that all the IDs in an $n$-node $H(n, \Delta)$-labeled graph are unique, which follows from the fact that the girth of $H(n, \Delta)$ is strictly larger than $n$.
\end{proof}

\paragraph{Hardness of Deterministic Sinkless Orientation relative to an ID graph}
\label{paragraph:det_so_lower_bound}
We now prove that there cannot be a deterministic local algorithm with volume complexity $o(n)$ that works with exponential IDs, even when those IDs satisfy constraints defined by an ID graph $H(n, \Delta)$.

To do so, we first show that an $o(n)$-probe \volume algorithm implies that there exists some constant $n^* \in \mathbb{N}$ and a deterministic \local algorithm that solves Sinkless Orientation on an (infinite) properly $H(n^*,\Delta)$-labeled tree in strictly less than $n^*$ rounds. This is an analogue of \cref{lem:exponential_id_speedup}.

\begin{lemma}
\label{lem:so_speedup}
If there exists a deterministic \volume algorithm $\fA$ that solves Sinkless Orientation on $n$-node trees with maximum degree $\Delta = O(1)$ that are properly $\Delta$-edge colored relative to $H(n, \Delta)$ and with a probe complexity of $f(n) \leq n/(3\Delta)$ for all large enough $n$, then there exists a constant $n^*$ and a deterministic \local algorithm $\fA'$ that solves Sinkless Orientation on all (possibly infinite) trees with maximum degree $\Delta$ that are properly $\Delta$-edge colored and $H(n^*, \Delta)$-labeled in fewer than $n^*$ rounds.
\end{lemma}
\begin{proof}
Consider running the algorithm $\fA$ for fixed $n^*$  and $\Delta$ such that $f(n^*) \leq n^*/(3\Delta)$ on \emph{any} tree $T$ having a maximum degree of $\Delta$ that is properly $\Delta$-edge colored and $H(n^*, \Delta)$-labeled. We now prove that $\fA$ solves the Sinkless Orientation problem on $T$. 
If not, then there exists a node $u$ such that either all the edges are oriented towards $u$ or $u$ has a neighbor $v$ such that the outputs of $u$ and $v$ are inconsistent. Consider now 
the set $S$ consisting of the at most $2 \cdot (n^*/3\Delta)$ vertices that $\fA$ probes in order to compute the answer for both $u$ and $v$. The set $S \cup N(S)$ has at most $\Delta \cdot |S| < n^*$ vertices. If we run $\fA$ on $G[S \cup N(S)]$, it will fail, too, since the two runs of $\fA$ on $T$ and $T[S \cup N(S)]$ are identical. 
Appending a non-zero number of vertices to  the vertices from $N(S)$ and labeling them so that the final labeling is still a $H(n^*, \Delta)$-labeling gives a graph $T'$ on exactly $n^*$ nodes that is properly $H(n^*, \Delta)$-labeled such that $\fA$ fails on $T'$. 
This is a contradiction with the correctness of $\fA$ on $n^*$-node trees. 

Hence, we get a deterministic \volume algorithm for Sinkless Orientation that performs at most $n^*/(3\Delta)$ probes and that works for any $H(n^*,\Delta)$-labeled tree $T$. This directly implies that there exists a \local algorithm with a round complexity of at most $n^*/(3\Delta)$ that solves Sinkless Orientation on any $H(n^*,\Delta)$-labeled tree $T$, as needed.
\end{proof}

Finally, the following theorem is a simple adaptation of the lower bound for Sinkless Orientation in the \local model via round elimination in \cite{brandt_LLL_lower_bound}). 
\begin{theorem}
\label{thm:so_hard_in_local}
The deterministic complexity of Sinkless Orientation in the \local model relative to $H(k, \Delta)$ is at least $k$. 
\end{theorem}
This theorem is proved in a parallel paper~\cite{brandt_chang_grebik_grunau_rozhon_vidnyaszky2021LCLs_on_trees_descriptive} developing the technique for a different purpose. For completeness we also leave a proof to \cref{sec:appendix}. 
Finally, we can put all the pieces together to prove \cref{thm:so_lower}. 

\begin{proof}[Proof of \cref{thm:so_lower}]
We first apply the derandomization of \cref{lem:so_rand_to_det} to deduce the existence of a deterministic algorithm with $o(n)$ probes, relative to $\{H(n, \Delta)\}_{n \in \mathbb{N}}$. 
Afterwards, we use \cref{lem:so_speedup} to conclude that there is an $r > 0$ and a deterministic local algorithm that solves sinkless orientation in less than $r$ rounds relative to an ID graph $H(r, \Delta)$, which is finally shown to be impossible by \cref{thm:so_hard_in_local}. 
\end{proof}
\begin{remark}
One can check that we actually prove the existence of some $\eps = \eps(\Delta) > 0$ such that any randomized LCA algorithm for sinkless orientation needs at least $\eps \log n$ probes for all $n \ge n_0$. 
\end{remark}

\section{Upper Bound for LLL on Constant Degree Graphs}

In this section, we complement the lower bound by proving the following matching upper bound result.
\begin{theorem}
\label{theorem:lll_upper}
There exists a fixed constant $c$ such that the randomized \lca /\volume complexity of the LLL on constant degree graphs under the polynomial criterion $p(e\Delta)^c \leq 1$ is $O(\log n)$.
\end{theorem}
\begin{proof}
The result follows by a slight adaptation of the \local algorithm of \cite{fischer2017sublogarithmic} to the \volume model. In particular, Fischer and Ghaffari show how to shatter a constant-degree graph in $O(\log^*n)$ rounds of the \local model. By shattering, we mean that their algorithm fixes the values for a subset of the random variables such that the following two properties are satisfied. 
\begin{enumerate}
    \item The probability of each bad event conditioned on the previously fixed random variables  is upper bounded by $\Delta^{-\Omega(c)}$.
    \item Consider the graph induced by all the nodes whose corresponding bad event has a non-zero probability of occurring. The connected components in this graph all have a size of $O(\log n)$ with probability $1 - 1/\poly(n)$. 
\end{enumerate}
Note that the shattering procedure -- denoted as the pre-shattering phase -- directly implies a randomized \volume algorithm with probe complexity $O(\log n \cdot  \Delta^{O(\log^*n})$. To see why, note that we can determine the state of all random variables after the pre-shattering phase that a given bad event depends on with a probe complexity of $\Delta^{O(\log^*n)}$ by applying the Parnas-Ron reduction to the $O(\log^* n)$ round \local algorithm. 
This in turn allows us to find the connected component of a given node with $O(\log n \cdot \Delta^{O(\log^* n)})$ probes, as long as the connected component has a size of $O(\log n)$, which happens with probability $1 - 1/\poly(n)$. 
Afterwards, one can find a valid assignment of all the random variables in the connected component in a brute-force centralized manner. Note that the standard LLL criterion $ep(\Delta + 1) \leq 1$ guarantees the existence of such an assignment. 

To improve the probe complexity to $O(\log n)$, we show how to adapt the pre-shattering phase such that it runs in $O(1)$ \local rounds while still retaining the two properties stated above. Once we have shown this, the aforementioned \volume simulation directly proves \cref{theorem:lll_upper}. The pre-shattering phase of \cite{fischer2017sublogarithmic} works by first computing a $2$-hop-coloring with $\Delta^2 + 1$ colors, where a $2$-hop coloring is a coloring that assigns any two nodes having a distance of at most $2$ a different color.
Once this coloring is computed, the algorithm iterates through the constantly many color classes and fixes in each iteration a subset of the random variables. Each random variable will be set with probability $1 - \Delta^{-\Omega(c)}$, even when fixing the randomness outside the $c_1$-hop neighborhood for some fixed constant $c_1 \geq 2$ not depending on $c$ adversarially. The only step that takes more than constant time is the computation of the coloring. Instead of computing the coloring deterministically, we instead assign each node one  out of $\Delta^{c'}$ colors for some fixed positive constant $c' \gg 1$, independently and uniformly at random. We say that a node fails if its chosen color is not unique in its $2$-hop neighborhood. Note that a given node fails with probability at most $1/\Delta^{\Omega(c')}$. We then postpone the assignment of each random variable that affects any of the failed nodes. For all the other random variables, we run the same process as described in \cite{fischer2017sublogarithmic} by iterating through the $\poly(\Delta) = O(1)$ color classes in $O(1)$ rounds of the \local model. The pre-shattering phase of \cite{fischer2017sublogarithmic} still deterministically guarantees that each bad event occurs with probability $1 - \Delta^{-\Omega(c)}$ conditioned on the random variables set in the pre-shattering phase. Moreover, each bad event that does not correspond to a failed node occurs with probability $0$ after the pre-shattering partial assignment with probability at least  $1 - \Delta^{-\Omega(c)}$, independent of the randomness outside the $c_1$-hop neighborhood.  An application of \cref{lem:shattering} therefore guarantees that by choosing $c$ and $c'$ large enough, each connected component after the pre-shattering phase has a size of $O(\log n)$ with high probability in $n$, thus concluding the proof of \cref{theorem:lll_upper}.

\begin{lemma}[The Shattering Lemma, cf. with Lemma 2.3 of \cite{fischer2017sublogarithmic}] 
\label{lem:shattering}

Let $G = (V,E)$ be a graph with maximum degree $\Delta = O(1)$. Consider
a process which generates a random subset $B \subseteq V$ where $Pr[v \in B] \leq \Delta^{-c_1}$
, for some constant $c_1 \geq 1$,
independent on the randomness of nodes outside the $c_2$-hop neighborhood of $v$, for all $v \in V$, for some constant $c_2 \geq 1$. Then, with probability at least $1 - n^{-c_3}$
, for any constant $c_3 < c_1 - 4c_2 - 2$, we have that each connected component of $G[B]$ has size $O(\log n)$.
\end{lemma}

\end{proof}

\section{Lower Bound for Constant Coloring}

\coloring*

\begin{proof}
The upper bound of $O(n)$ follows trivially from the fact that every tree is bipartite. To prove the lower bound, assume for the sake of contradiction that there exists a deterministic \volume algorithm $\mathcal{A}$ that $c$-colors a bounded degree tree with $o(n)$ probes. By a result from Bollobás \cite{BOLLOBAS1978311}, there exists a (connected) graph $G$ on $n$ nodes with chromatic number strictly greater than $c$ that has a constant maximum degree $\Delta_G$ (with the constant only depending on $c$) and girth $g = \Omega(\log_c n)$. Now, we consider the unique infinite $\Delta_H$-regular graph $H$ (up to isomorphism) that contains $G$ as an induced subgraph and $G$ and $H$ have the same set of cycles. We choose $\Delta_H$ as small as possible such that $(\Delta_H - 1)^{g/4} \geq n^{10}$. Note that $\Delta_H = \poly(\Delta_G)$ and therefore $H$ has constant maximum degree. Now, we assign each node in $H$ an identifier uniformly and independently at random from the set $\{1,2,\ldots,n^{10}\}$. Note that these identifiers are not unique. Moreover, we also randomize the port assignment in $H$ such that each node chooses its port assignment independently and each permutation has the same probability to be chosen. Now, for every query corresponding to a node in $G$, we run algorithm $\mathcal{A}$ on $H$ to provide an answer to the query. Even though $H$ has an infinite number of vertices and contains cycles, we tell $\mathcal{A}$ that it is a tree with exactly $n$ vertices. Note that the range of identifiers supports that illusion, though the algorithm might encounter two nodes with the same ID or detect a cycle. We now need the following claim.

\begin{claim*}
	With strictly positive probability over the randomness of the ID-and port assignment, all nodes probed by $\mathcal{A}$ while answering the $n$ different queries got assigned pair-wise distinct IDs. Moreover, while answering the query for some node $v$, $\mathcal{A}$ does not probe a node $u$ corresponding to a node in $G$ such that the distance between $u$ and $v$ is at least $g/4$.
\end{claim*}

Before we prove this claim (in Lemma~\ref{lem:no_far_vertex}), we show how it implies \cref{theorem:constant_coloring_theorem}.
According to the probabilistic method, there exists a fixed assignment of identifiers and ports such that when we run the process described above with this fixed assignment, there do not exist two distinct nodes that got assigned the same ID such that $\mathcal{A}$ probed these two nodes at some point in the process of answering the $n$ different queries, and for each node $v$ that $\mathcal{A}$ is queried on, $\mathcal{A}$ does not probe a node $u$ corresponding to a node in $G$ during answering the query for $v$ such that the distance between $u$ and $v$ is strictly greater than $g/4$.
Now, let $v$ and $w$ be two arbitrary neighboring nodes in the graph $G$. Consider the graph induced by all the nodes that $\mathcal{A}$ has probed in the graph $H$ when answering the query for $v$ and $w$.
Note that this induced graph does not contain any cycle, as otherwise when queried on $v$ or $w$, algorithm $\mathcal{A}$ would have seen a node in $G$ with a distance strictly greater than $g/4$. Hence, the induced graph is a bounded-degree forest with $o(n)$ vertices and all the vertices have a unique identifier.
Hence, for $n$ large enough, we can add additional vertices and edges to make it a bounded-degree tree $T_{v,w}$ on $n$ vertices such that $T_{v,w}$ would be a valid input to the algorithm and when queried on $T_{v,w}$, the colors that $\mathcal{A}$ outputs at $v$ and $w$ would still be the same as when queried on $H$. Note that here we make use of the fact that $\mathcal{A}$ is deterministic, as otherwise we would not have the guarantee that $\mathcal{A}$ would probe the exact same vertices in the exact same order.
Now, as we assumed that $G$ has a chromatic number strictly greater than $c$, there are two neighboring nodes $v, w$ in $G$ that get assigned the same color by $\mathcal{A}$. As these two nodes are also neighboring in $T_{v,w}$, $\mathcal{A}$ outputs the same color for two neighboring nodes in $T_{v,w}$. This is however a contradiction as $T_{v,w}$ is a valid input for $\mathcal{A}$ and we assumed that $\mathcal{A}$ is a correct deterministic algorithm.
\end{proof}

It remains to prove the claim in the proof of Theorem~\ref{theorem:constant_coloring_theorem}, which we restate as Lemma~\ref{lem:no_far_vertex}.
\begin{lemma}
\label{lem:no_far_vertex}
	With strictly positive probability over the randomness of the ID-and port assignment, there do not exist two distinct nodes that got assigned the same ID such that $\mathcal{A}$ probed these two nodes at some point in the process of answering the $n$ different queries. Moreover, while answering the query for some node $v$, $\mathcal{A}$ does not probe a node $u$ corresponding to a node in $G$ such that the distance between $u$ and $v$ is at least $g/4$.
\end{lemma}
\begin{proof}
First, note that $\mathcal{A}$ performs at most $n^2$ probes to answer the $n$ queries and therefore sees at most $n^2$ different nodes. For $i \neq j \in [n^2]$, let $A_{ij}$ be the event that the $i$-th and the $j$-th vertex that $\mathcal{A}$ probes are distinct vertices and have the same ID. We have $Pr[A_{ij}] = 1/n^{10}$ and thus a union bound over the $n^4$ pairs implies that the first part of the lemma holds with probability $1- 1/n^6$.

Before showing that the second part of the lemma holds with probability at least $1/2$, we first give some intuition why it holds. Let $v$ be a node in $H$ corresponding to a node in $G$. As the girth of $H$ is $g$ and $H$ is $\Delta_H$-regular, the $g/4$-hop neighborhood around $v$ is a tree with each node having a degree of $\Delta_H$ except the leaf vertices. Note that the total number of leaf vertices is at least $n^{10}$, but at most $n$ of them correspond to nodes in $G$. Hence, it is intuitively hard for $\mathcal{A}$ to probe any of the leaf vertices corresponding to nodes in $G$. However, it needs to probe such a vertex in order to find a node in $G$ having a distance strictly greater than $g/4$. Though the intuition is simple, one needs to carefully argue that this intuition is correct. To do so, we perform a series of reductions. For the sake of contradiction, assume that $\mathcal{A}$ finds with probability at least $1/n^2$ a vertex in $G$ that has a distance of $g/4$ to a given queried vertex that corresponds to a node in $G$. 

\paragraph{Reduction 1: Omitting the Identifiers}
First, we argue that $\mathcal{A}$ does not really need the random identifiers. To do so, consider a \volume model variant in which the algorithm itself assigns IDs to vertices. That is, each time the algorithm probes a vertex it has not encountered before it can assign this vertex a new ID. Now, consider the following randomized algorithm $\mathcal{A}'$ that works in this new setting\ by assigning each newly encountered vertex an identifier uniformly and independently at random from the set $\{1,2,\ldots, n^{10}\}$ and then simulating algorithm $\mathcal{A}$ with these self-assigned IDs. Then, both $\mathcal{A}$ and $\mathcal{A}'$ have exactly the same probability of finding a vertex that corresponds to a vertex in $G$ and that has a distance of at least $g/4$ to the queried vertex $v$. 

\paragraph{Reduction 2: No Probes Outside the $g/4$-hop neighborhood}
For the next reduction, we consider the setting with self-assigned identifiers as above and add the following modification. Once an algorithm encounters some vertex with a distance of exactly $g/4$ to the originally queried vertex $v$, we tell the algorithm whether this vertex corresponds to a vertex in $G$ or not. Furthermore, we do not allow any probes outside the $g/4$-hop neighborhood around $v$. Thus, the algorithm only gets to know vertices inside the $g/4$-hop neighborhood of $v$. Now, we explain how to turn a randomized algorithm that works in the setting with the self-assigned IDs into an algorithm that has the exact same probability of finding a far-away vertex that corresponds to a vertex in $G$ in this new setting described above. We again simulate the randomized algorithm. If the algorithm encounters a vertex that corresponds to a vertex in $G$  with a distance of $g/4$ to $v$, then the algorithm simply stops the execution as the algorithm has achieved its goal. Otherwise, if the algorithm would like to make a probe outside the $g/4$-hop neighborhood, then the algorithm \emph{virtually} simulates these probes. The structure of the graph outside the $g/4$-hop neighborhood (not counting the vertices that correspond to nodes in $G$) is completely known to the algorithm without making any probes. The only unknown are the port assignments. However, the algorithm can make up a random port assignment in its mind (with exactly the same distribution as the port assignment would have in the real graph) and then answer the probes accordingly. It is easy to see that this new algorithm succeeds with the same probability as the previous algorithm.

\paragraph{Reduction 3: Guessing Game}
For a given port assignment of $H$, we define an order on all the nodes of distance precisely $g/4$ from $v$ in the following way: We associate with each such node the sequence of ports one needs to take to get from $v$ to the respective vertex. Then, we order the vertices according to the lexicographical order of their associated sequences. Let $N_{g/4} \geq n^{10}$ be the number of nodes with a distance of $g/4$ to $v$. With each port assignment $p$ of the $g/4$-hop neighborhood around $v$, we identify a tuple $\tup_p \in \{0,1\}^{N_{g/4}}$. We set $(\tup_p)_i = 1$ if the $i$-th vertex with respect to the order defined corresponds to a vertex in $G$ and 0 otherwise. Thus, $\tup_p$ contains at most $n$ non-zero entries.

Consider the distribution over $\{0,1\}^{N_{g/4}}$ that we get by choosing the port assignment at random as described in the beginning.
Now, we define the following game: A port assignment $p$ is chosen uniformly at random. This port assignment is generally unknown to the algorithm. The only information the algorithm obtains about the port assignment is, for each vertex, the port number corresponding to the edge leading to the parent of the vertex. By parent, we mean the parent with respect to the tree rooted at $v$ induced by the $g/4$-hop neighborhood of $v$. Now, the task of the (randomized) algorithm is to output an index set $I \subset [N_{g/4}]$ with $|I| \leq n$.
The algorithm is said to win the game if there exists an index $\ind \in I$ such that $(\tup_p)_{\ind} = 1$. 

We show that we can construct a (randomized) algorithm that wins the game with a probability of at least $1/n^2$ when given the algorithm from the previous setting (which finds a vertex that corresponds to a node in $G$ and that has a distance of $g/4$ to $v$ with probability of at least $1/n^2$) as a subroutine.
To do this, we  simulate the algorithm from the previous subsection. Then, we take $I$ as the index set corresponding to all the vertices of distance $g/4$ to $v$ that the algorithm ``finds'' during the execution. Note that in order to determine the corresponding index set, we need to know the position of any examined vertex in the lexicographical order.
However, this can be done as it is sufficient to know the port assignments of all the parents. To be able to simulate the algorithm from the previous section we need to provide the neighborhood oracle corresponding to the chosen port assignment $p$.
Although $p$ is not completely known, this can actually be done as it is sufficient to just know the port number corresponding to the parent. Furthermore, we need to answer whether a vertex of distance $g/4$ that the algorithm probes corresponds to a vertex in $G$ or not.
This is actually not possible. However, it is sufficient to always tell the algorithm that the vertex does not correspond to a node in $G$. Now, we show that we win the guessing game with probability at least $1/n^2$. This can be seen as follows: Let $p$ be the selected port assignment function and $r$ the random bits that the simulated algorithm uses.
Assume that the algorithm loses the game. Now, consider the simulated algorithm is run in the previous setting with the same port assignment function and the same random bits. Then, the algorithm would not have found a vertex of distance $g/4$ to $v$ that corresponds to a node in $G$. This follows as the execution is identical in both cases as the oracle calls and the answers to whether a vertex of distance $g/4$ to $v$ corresponds to a node in $G$ is answered in the same way in both cases.

\paragraph{Guessing Game is Impossible}
Next, we show that the game that we have defined above cannot be won with a probability of at least $1/n^2$. To that end, consider some fixed index set $I = \{\ind_1,..., \ind_k\}$ with $k \leq n$. For $i \in [k]$, we define $A_i$ as the event that the sequence $\seq \in \{0,1\}^{N_{g/4}}$ chosen in the game satisfies $\seq_{\ind_i} = 1$. Also note that $A_i$ is independent from the information given to the algorithm, i.e., the collection of port numbers corresponding, for each vertex, to the edge leading to its parent. By symmetry, we have

\[
    Pr[A_i] \leq \frac{n}{N_{g/4}} \leq \frac{1}{n^{9}}.
\]

Furthermore, we define $A = \bigcup_{i \in I} A_i$. $A$ corresponds exactly to the event that one wins the game if one had chosen the index set $I$. Thus, we can use a union bound to get:

$$Pr[A] = Pr[\bigcup_i A_i] \leq \sum_{i} Pr[A_i] \leq k \cdot \frac{1}{n^9} \leq \frac{1}{n^8}.$$

Hence, we arrived at a contradiction, which proves \cref{lem:no_far_vertex} (and therefore \cref{theorem:constant_coloring_theorem}).
\end{proof}

\section*{Acknowledgements}
We thank Mohsen Ghaffari, Jan Grebík, and Jukka Suomela for useful discussions. 
This project has received funding from the European Research Council (ERC) under the European
Unions Horizon 2020 research and innovation programme (grant agreement No. 853109).

\bibliographystyle{alpha}
\bibliography{ref}

\appendix

\section{Missing proofs from \cref{subsec:id_graph}}
\label{sec:appendix}

Here we prove \cref{lem:id_graph_existence}.
\begin{proof}
In the following, we assume that both $\Delta$ and $R$ are sufficiently large. Furthermore, we define $n := \Delta^{1000R}$. We let each $H_i$ be equal to an Erd\H{o}s-R\'enyi graph with $n$ vertices and where each edge is included with probability $p = \Delta^2 / n$.

The expected number of cycles of length less than $10R$ in $H$ is upper bounded by
\[
\sum_{j = 2}^{10R-1} n^j (\Delta p)^j = \sum_{j = 2}^{10R-1} \left(\Delta^3 \right)^j \leq \Delta^{10R} \leq n^{1/100}.
\]

Let $V_{cycle}$ denote the set of all vertices that are contained in a cycle of length less than $R$ in $H$. It holds that $|V_{cycle}| \leq n^{1/2}$ with probability at least $99/100$. 
Next, we denote with $V_{deg}$ the set of all vertices whose degree in some $H_i$ is $0$ or the degree is at least $\Delta^{10}$ in $H$. As the expected degree of each node in $H_i$ is at least $0.5\Delta^2$ and the expected degree of each node in $H$ is at most $\Delta^3$, a Chernoff Bound followed by a union bound implies that a given vertex $v \in V(H)$ is contained in $V_{deg}$ with probability at most $\Delta \cdot e^{-\Theta(\Delta^2)}$. Hence, the expected size of $V_{deg}$ is at most $ n \cdot \Delta \cdot e^{-\Theta(\Delta^2)}$. Thus, with probability at least $\frac{99}{100}$ it holds that $|V_{deg}| \leq \frac{n}{\Delta^{10}}$. \\
Next, we bound the probability that there exists a subset $S$ consisting of $\lceil \frac{2n}{\Delta^{10}} \rceil$ many vertices such that

\[\sum_{v \in S} deg_H(v) \geq \frac{n}{\Delta^{2}}.\]

For a fixed set $S$, the expected value of $\sum_{v \in S} deg_H(v)$ is at most

\[\frac{\Delta^3}{n}\binom{|S|}{2} \leq \frac{n}{2\Delta^{2}}.\]

Hence, by a Chernoff Bound followed by a union bound, the probability that we have such a set $S$ is at most 

\[\binom{n}{\lceil 2n/\Delta^{10} \rceil} \cdot e^{-\Theta(n/\Delta^{2})} \leq (e\Delta^{10})^{\lceil 2n/\Delta^{10} \rceil} \cdot e^{-\Theta(n/\Delta^{2})} \leq \frac{1}{100}.\]

Next, we bound the probability that the size of the largest independent set in $H_i$ is at least $\lceil n/\Delta^{3/2} \rceil$ for some $i \in [\Delta]$. We can upper bound the probability by

\begin{align*}
    \Delta \cdot \binom{n}{\lceil n/\Delta^{3/2} \rceil} \cdot \left(1 - \Delta^2  /n\right)^{\binom{\lceil n/\Delta^{3/2} \rceil}{2}}  &\leq \Delta \left(e \Delta^{3/2}\right)^{\lceil n/\Delta^{3/2} \rceil} e^{-\Theta(n/\Delta)} \\
    &\leq e^{\Theta(\log(\Delta)n/\Delta^{3/2})} e^{-\Theta(n/\Delta)} \\
    &= e^{-\Theta(n/\Delta)} \\
    &\leq \frac{1}{100}.
\end{align*}

Now, let $V_{rem} := V_{deg} \cup V_{cycle}$ and denote with $V_{fix}$ the set of all vertices that have at least one neighbor in $V_{rem}$. We can assume the following.

\begin{itemize}
    \item $|V_{rem}| \leq \frac{2n}{\Delta^{10}}$
    \item $|V_{fix}| \leq \frac{n}{\Delta^2}$ 
    \item For each $i \in [\Delta]$, the size of the largest independent set in $H_i$ is at most $\frac{n}{\Delta^{3/2}}$.
\end{itemize}

Now, let $H',H_1',\ldots,H_\Delta'$ denote the graphs obtained from $H,H_1,\ldots,H_\Delta$ by removing all the vertices in $V_{rem}$. The girth of $H'$ is at least $R$, the maximum degree of $H'$ is at most $\Delta^{10} - 1$ and $H'$ has at least $n/2$ vertices. Now, let $V_0$ denote the set of vertices that have a degree of $0$ in one of the $H'_i$'s. As $V_0 \subseteq V_{fix}$, we can deduce that $|V_0| \leq |V_{fix}|$. Next, we iteratively do the following: as long as there exists some vertex $v \in V(H')$ and some $i \in [\Delta]$ such that $deg_{H_i'}(v) = 0$, we add one edge incident to $v$ to $H_i'$ (and therefore $H'$) such that the girth of $H'$ is still at least $10R$ and the maximum degree of $H'$ is at most $\Delta^{10}$. Note that we add at most $|V_0|\Delta \leq \frac{n}{\Delta}$ edges in total. To show that we can always add such an edge, note that there are at most $(\Delta^{10})^{10R + 1} \leq \sqrt{n}$ vertices with a distance less than $R$ to $v$. Hence, there are at least $n/2 - \sqrt{n} - \frac{n}{\Delta}$ vertices with a degree of at most $\Delta^{10} - 1$ in $H'$ and a distance of at least $10R$ to $v$. Hence, we can add such an edge to $H'_i$ (and therefore $H'$) such that $H'$ still has a maximum degree of at most $\Delta^{10}$ and a girth of at least $10R$. Now, let $H'',H_1'',\ldots,H_\Delta''$ denote the graphs obtained after we have added all the edges. The resulting graph $H''$ has at least $n/2$ vertices and the girth of $H''$ is at least $10R$. Moreover, the degree of each vertex in $H_i''$ is at least $1$ and the size of the largest independent set in $H_i''$ is at most $\frac{n}{\Delta^{3/2}} < \frac{|V(H'')|}{\Delta}$.  Hence, $H_1'', H_2'', \ldots, H_\Delta''$ satisfy our desired properties.

\end{proof}

Next, we prove \cref{thm:so_hard_in_local}. The proof follows along the lines of \cite{brandt_LLL_lower_bound}. It is in fact simpler as one does not need to keep track of probabilities.
\begin{proof}
Assume that there is a $t$-round \local algorithm $\fA$ with $t \le k$ that solves Sinkless Orientation on an infinite $\Delta$-regular tree that is properly $\Delta$-edge colored and $H(k, \Delta)$-labeled.
We show that then there is also a $(t-1/2)$-round \local algorithm $\fA'$ for the same problem, that is, an algorithm where the decision of each edge $e = (u,v)$ depends on $B(e,t-1) := B(u,t-1) \cup B(v, t-1)$. The algorithm $\fA'$ considers all possible ways how $B(u, t) \setminus B(e,t-1)$ can be extended with an $H(k,\Delta)$-labeling and if there is one such extension such that $\fA$ orients the half-edge $(u,e)$ out, then $\fA'$ orients $e$ in the direction from $u$ to $v$. The same holds for $v$ and if neither case occurs, $e$ is oriented arbitrarily. Observe that it cannot happen that both $B(u,t) \setminus B(e,t-1)$ and $B(v,t) \setminus B(e, t-1)$ can be extended with an $H(k,\Delta)$-labeling such that $\fA$ orients both $(u,e)$ out and $(v,e)$ out as this would mean that $\fA$ is not correct. Here we importantly use the fact that a valid $H$-extension for $B(u, t) \setminus B(e,t-1)$ and a valid $H$-extension for $B(v, t) \setminus B(e,t-1)$ can be naturally ``glued'' together so as to yield one valid $H$-extension. 
An analogous reasoning allows us to transform $\fA'$ into a $(t-1)$-round algorithm $\fA''$. 

Repeating the above reasoning we conclude that a $0$-round algorithm $\fA^*$ exists. Such an algorithm decides for each vertex the orientation of its half-edges only based on its $H(k,\Delta)$-label. 
As for each vertex $\fA^*$ needs to orient at least one of its half-edges in the outward direction, we can color each vertex of $H(k,\Delta)$ with a color from $[\Delta]$ such that $\fA^*$ orients the edge with the respective color outwards. By the pigeonhole principle, there exists a color $c \in [\Delta]$ such that at least a $(1/\Delta)$-fraction of vertices of $H$ is colored with $c$. However, \cref{def:id_graph} implies that the set of vertices colored by $c$ in $H(k,\Delta)$ is not independent in $H_c$. Hence, we get an example of a two-node configuration where $\fA^*$ fails, a contradiction. 
\end{proof}

\end{document}